\documentclass[journal, 10pt]{IEEEtran}
\usepackage{graphicx} 
\usepackage[utf8]{inputenc} %
\usepackage[T1]{fontenc}    %
\usepackage{hyperref}       %
\usepackage{url}            %
\usepackage{booktabs}       %
\usepackage{amsfonts}       %
\usepackage{nicefrac}       %
\usepackage{microtype}      %
\usepackage{caption}
\usepackage{microtype}
\usepackage{graphicx}
\usepackage{subcaption}
\usepackage{booktabs} %
\usepackage{enumitem} 
\usepackage{bm}
\usepackage{placeins}
\usepackage{algorithmicx,algpseudocode,algorithm}
\usepackage{wrapfig,tikz}
\usepackage{amsmath,longtable,fancyhdr,booktabs,multirow,graphicx,float}
\usepackage{adjustbox, bigstrut, tabularx, multirow, makecell, diagbox}
\usepackage{amssymb,xcolor,amsthm}
\usepackage{color}
\usepackage{colortbl}
\usepackage{theoremref}
\usepackage{stmaryrd}
\usepackage{xcolor}
\usepackage{amsthm}

\newtheorem{theorem}{Theorem}[section]
\newtheorem{lemma}[theorem]{Lemma}
\newtheorem{definition}[theorem]{Definition}

\newtheorem{assumption}[theorem]{Assumption}
\usepackage{cuted}

\ifCLASSINFOpdf
\else
\fi
\begin{document}
%
\title{Score-Based Quickest Change Detection and Fault Identification for Multi-Stream Signals}
%
%
%

\author{Wuxia~Chen, Sean~Moushegian, 
Vahid~Tarokh,~\IEEEmembership{Fellow,~IEEE,  } Taposh~Banerjee}
\maketitle

\begin{abstract}
This paper introduces an approach to multi-stream quickest change detection and fault isolation for unnormalized and score-based statistical models. Traditional optimal algorithms in the quickest change detection literature require explicit pre-change and post-change distributions to calculate the likelihood ratio of the observations, which can be computationally expensive for higher-dimensional data and sometimes even infeasible for complex machine learning models. To address these challenges, we propose the min-SCUSUM method, a Hyvärinen score-based algorithm that computes the difference of score functions in place of log-likelihood ratios.
We provide a delay and false alarm analysis of the proposed algorithm, showing that its asymptotic performance depends on the Fisher divergence between the pre- and post-change distributions. Furthermore, we establish an upper bound on the probability of fault misidentification in distinguishing the affected stream from the unaffected ones. 
\end{abstract}


%
\IEEEpeerreviewmaketitle

\section{Introduction}
Detecting abrupt changes in the statistical characteristics of sequentially acquired data is a critical problem that spans a wide spectrum of applications, including power system transmission outage \cite{chen2015quickest}, quality control \cite{hawkins2003changepoint},  target identification \cite{blackman2004multiple},  biology \cite{siegmund2013change}, and more.  In statistics, this problem is formulated as the problem of quickest change detection (QCD), where the goal is to detect a change in the distribution of a sequence of random variables as swiftly as possible while controlling the rate of false alarms \cite{shiryaev1963optimum}.  

One of the most prominent algorithms for QCD is the Cumulative Sum (CUSUM) algorithm \cite{page1954continuous}, which is particularly effective in the context of independent and identically distributed (i.i.d.) data before and after the change. The CUSUM algorithm is not only practical due to its recursive structure but also theoretically optimal in minimizing the worst-case detection delay, as described by Lorden's criterion \cite{lorden1971procedures}, while adhering to a predefined false alarm rate \cite{moustakides1986optimal}. For an overview of QCD and other optimal algorithms, we refer to 
\cite{veeravalli2014quickest, tartakovsky2014sequential, tartakovsky2019sequential, poor2008quickest, basseville1993detection}. 

The CUSUM and other optimal algorithms from the QCD literature depend explicitly on knowing both pre- and post-change distributions, a task that becomes computationally intensive with high-dimensional data or when explicit distributions are not readily available. For instance, in machine learning applications involving energy-based models \cite{lecun2006tutorial} or score-based deep generative models \cite{song2019generative}, computing the necessary likelihood ratios for these algorithms can be particularly demanding. 
To address this issue, the authors in \cite{Wuetal_IT_2024, Banerjee_SQA_2024} developed a score-based approach to QCD using the Hyvärinen score of the probability densities. These score-based algorithms can be applied to data where the densities 
are available in unnormalized or score-based forms, which is the situation frequently encountered in machine learning applications. 

In many real-world applications, it is not only important to detect changes quickly but also to accurately identify the nature of the change. This dual objective leads to the problem of sequential change diagnosis, where the goal is to determine the correct post-change distribution among several possibilities once a change has been detected. This problem has also been studied extensively in the literature \cite{nikiforov1995generalized, nikiforov2002simple, nikiforov2003lower, nikiforov2016sequential, tartakovsky2014sequential, tartakovsky2019sequential}, often leading to algorithms that are computationally hard to implement. In a recent work \cite{warner2024worst}, however, the authors have shown that a popular and computationally efficient min-CUSUM algorithm is also asymptotically optimal for a well-defined problem formulation. 
 In the min-CUSUM algorithm, when there are a finite number of post-change alternatives, the CUSUM statistic is calculated for each alternative. A change is declared when any one of the CUSUM statistics crosses a threshold. The post-change distribution corresponding to the CUSUM statistic that triggered the alarm is identified as the true post-change distribution.

While the min-CUSUM algorithm is computationally efficient and optimal, it cannot be applied to high-dimensional data, 
where the precise form of the distributions is not available. Consider the problem of anomaly detection using parallel video streams. Here, the goal is to detect an anomaly in the sequence of images in each camera. Once an anomaly has been detected, we need to identify the video stream where the anomaly was seen. This problem has applications, for example, in military surveillance and traffic data. Because of the high-dimensional nature of the images, the statistical density of the images is intractable to learn. As a result, the min-CUSUM algorithm from \cite{warner2024worst} cannot be applied to this problem. 

In this paper, we propose the min-SCUSUM algorithm, a Hyvärinen score-based modification of the min-CUSUM algorithm, for QCD and fault isolation in multi-stream data. The min-SCUSUM algorithm is obtained by replacing the log likelihood ratios in the min-CUSUM algorithm by differences of Hyvärinen scores (see Section~\ref{sec:min_SCUSUM_diagnosis}). 
The Hyvärinen scores can be learned directly from data using neural networks \cite{song2019generative} \cite{vincent2011connection}. In this paper, we also analyze the delay, false alarm, and fault identification performance of the min-SCUSUM algorithm. We note that many aspects of our analysis are novel since the classical methods of analysis employed in \cite{warner2024worst} must be modified because the min-SCUSUM algorithm is no longer based on likelihood ratios. Finally, we apply our algorithm to an anomaly detection application for video data.

To summarize, the main contributions of our paper are as follows:
\begin{itemize}
    \item We propose the min-SCUSUM algorithm, a new multi-stream QCD algorithm designed for quickest detection and fault identification in unnormalized and score-based models. The algorithm replaces the negative log-likelihood terms in CUSUM with a multiple of the differences of Hyvärinen scores.
    \item We provide a detailed analysis of the delay and false alarm performance of the min-SCUSUM algorithm in the i.i.d. setting, showing that it achieves asymptotic optimality under certain conditions.
    \item We give an upper bound on the probability of mis-identification given that there is no false alarm, and further validate the bound in a numerical experiment using synthesized high-dimensional data, where the empirical misidentification probabilities remain below the theoretical upper bound. 
    \item Finally, we apply our algorithm to detect changes in multi-channel video streams generated from real-world videos, demonstrating its effectiveness in a high-dimensional data environment.  While there can be several ways to detect and isolate anomalies in video data, we remark that our proposed algorithm works for arbitrary high-dimensional data, not just image data. Moreover, we provide an algorithm with provable performance guarantees on delay, false alarm rate, and the probability of mis-identification. 
\end{itemize}

The rest of the paper is organized as follows: 
In Section~\ref{sec:problem_formulation}, we formulate the quickest change detection and the fault isolation problem. 
In Section~\ref{sec:Hyvärinen_score}, we review the Hyvärinen score and the score-based change detection algorithm from \cite{Wuetal_IT_2024}. In Section ~\ref{sec:min_SCUSUM_diagnosis}, we introduce the min-SCUSUM algorithm, an algorithm for score-based fault detection and identification. In Section ~\ref{sec:FA_Delay_properties}, we analyze the delay, false alarm, and mis-identification performance of the min-SCUSUM algorithm. In Section~\ref{sec:numerical}, we apply the algorithm to real video data and simulated high-dimensional data to verify its effectiveness.

\subsection{Related Work}
Quickest change detection (QCD) in multi-stream settings has been extensively studied. In \cite{xie2013sequential},  a mixture-based generalized likelihood ratio procedure is proposed for detecting sparse mean shifts. In \cite{wang2015large}, the authors combined hard thresholding with shrinkage estimation to handle both sparse and dense scenarios. In another work  \cite{fellouris2017multistream}, an asymptotic framework has been developed for multichannel detection with growing sensor counts. In \cite{xu2021optimum}, a myopic sampling strategy for sequential detection under sensing constraints is studied. Other related works \cite{xu2021multi, xu2022active, oleyaeimotlagh2023quickest} have contributed sampling-aware and adaptive procedures. 
However, all these studies assume that the distributions are precisely known, and the methods there are based on likelihood ratios. 

To overcome the limitations of likelihood-based methods for high-dimensional distributions, recent works \cite{Wuetal_IT_2024} and  \cite{Banerjee_SQA_2024} use a Hyvärinen score-based approach for QCD in a single stream. The Hyvärinen score \cite{hyvarinen2005estimation} can be computed for unnormalized and score-based models  \cite{song2019generative, vincent2011connection}. In this paper, we extend the work in \cite{Wuetal_IT_2024} to a multi-stream scenario and perform QCD and fault isolation. Our approach to QCD with fault isolation is classical \cite{chen2015quickest, tartakovsky2014sequential}, but we are motivated by the recent optimality theory developed in \cite{warner2024worst}. Sequential fault diagnosis—identifying the affected stream after a change—has also been addressed in \cite{nikiforov1995generalized, nikiforov2002simple, nikiforov2003lower, nikiforov2016sequential}. However, the optimal tests are based on pairwise likelihood ratios and hence the algorithms are computationally demanding. 

Other nonparametric approaches have also been developed for high-dimensional or structured data. The authors in  \cite{chu2022sequential} proposed graph-based sequential detection for non-Euclidean observations, while \cite{banerjee2018quickest} introduced a method based on maximal kNN coherence. These works relax distributional assumptions, but do not address fault identification or score-based modeling.

\section{Problem Formulation} \label{sec:problem_formulation}
In this problem, there are $d \in \mathbb{N}$ independent streams or channels. At each discrete time $n \in \mathbb{N}$, we observe $X_{1,n}, \dots, X_{d,n}$ simultaneously, where $X_{i,n} \in \mathbb{R}^m$ represents the high-dimensional observation from channel $i$. The random variables are defined on a probability space $(\Omega, \mathcal{F}, P_{\nu,j})$. 
Under the measure $P_{\nu,j}$, the random variables satisfy the following change-point model:
\begin{equation}
\label{eq:CPmodel1}
    \begin{split}
        X_{i, n} &\sim f_i, \quad i \neq j, \; n \geq 1, \\
        X_{j, n} &\sim f_j, \quad n \leq \nu, \\
        X_{j, n} &\sim g_j, \quad n  > \nu.
    \end{split}
\end{equation}
The random variables are independent conditioned on the change point. Thus, under the measure $P_{\nu, j}$, the density of observations in only the stream $j$ changes from $f_j$ to $g_j$. This affected stream is unknown. Our goal is to detect this change as soon as it occurs and also to identify the affected stream. 
We denote by $P_\infty$ the probability measure under which no change occurs (i.e., $\nu = \infty$). In this case, we have
\begin{equation}
\label{eq:CPmodel1}
    \begin{split}
        X_{i, n} &\sim f_i,  \quad 1 \leq i \leq d,  \; n \geq 1, \\
    \end{split}
\end{equation}

We denote by $\mathbb{E}_\infty$ the expectation under $P_\infty$, and by $\mathbb{E}_{\nu,j}$ the expectation under $P_{\nu,j}$. Furthermore, we introduce the shorthand notations $P_j \triangleq P_{0,j}$ and $\mathbb{E}_j \triangleq \mathbb{E}_{0,j}$. Finally, we define the index set $\mathcal{I} = \{1, \dots, d\}$.

To state the problem objective mathematically, our objective is to devise a stopping time $T$ and a diagnosis $D \in \mathcal{I}$, both adapted to the filtration $\mathcal{F}_n$ with $\mathcal{F}_n = \sigma(X_1, \dots, X_n)$. On the event $\{T = n, D = i\} \in \mathcal{F}_n$, the change is declared at time $n$ while identifying channel $i \in \mathcal{I}$ as the one affected due to the change. 
Furthermore, since each observation $X_{i,n} \in \mathbb{R}^m$ is high-dimensional, with $m \gg 1$, we assume that we do not precisely know the pre-change densities $\{f_j\}$ and the post-change densities $\{g_j\}$. 
The stopping time $T$ optimizes the trade-off between delay and false alarm.  We define the set $\mathcal{C}$ as the collection of all diagnosis procedures $(T, D) $ \cite{tartakovsky2014sequential}.

To evaluate how well a diagnosis procedure $(T, D) \in \mathcal{C}$ performs in preventing false alarms, we consider the average number of observations 
until a stopping under pre-change measure $P_\infty$, denoted by $\mathbb{E}_\infty[T]$. We define $\mathcal{C}(\alpha)$ as the subset of diagnosis procedures that ensure the expected stopping time $T$ under $P_\infty$ is at least $\frac{1}{\alpha}$,  where $\alpha \in (0,1)$ is the maximum tolerable false alarm rate:
\begin{equation}
\mathcal{C}(\alpha) \triangleq \left\{(T, D) \in \mathcal{C} : \mathbb{E}_\infty[T] \geq \frac{1}{\alpha}  \right\}.
\end{equation}

To evaluate a diagnosis procedure $(T, D) \in \mathcal{C}$ in terms of its ability to correctly identify the affected stream once a detection alarm is raised, we consider the conditional probability of incorrect identification given that no false alarm has occurred, denoted as $P_{\nu,j} \left( D \neq j \mid T > \nu \right)$. We define $\beta$ as the tolerance for the worst-case false identification given no false alarm \cite{warner2024worst},  $$\max_{j \in \mathcal{I}} \sup_{\nu \in \mathbb{N}} P_{\nu,j} \left(D \neq j \mid T > \nu \right) \leq \beta. $$
We define $\mathcal{C}(\alpha, \beta)$ as the subset of diagnosis procedures within $\mathcal{C}(\alpha)$ that controls the worst-case conditional probability of false identification to be no greater than $\beta$:
\begin{equation}
\begin{aligned}
 \mathcal{C}(\alpha, \beta) &\triangleq
 \left\{ (T, D) \in \mathcal{C}:\, \mathbb{E}_\infty[T] \geq \frac{1}{\alpha}, \right. \\
 &\quad \quad \quad\left.   \max_{j \in \mathcal{I}} \sup_{\nu \in \mathbb{N}} P_{\nu,j} \left(D \neq j \mid T > \nu \right) \leq \beta \right\}.
\end{aligned}
\label{eq:T_D_constraint}
\end{equation}

Moreover, to measure the efficiency of a diagnosis procedure $(T, D) \in \mathcal{C}$ in quickly detecting a change when the post-change distribution is $P_i$ for some $i \in \mathcal{I}$, we use Lorden's criterion \cite{lorden1971procedures}. Here, we consider the worst-case conditional expected detection delay, accounting for both the change point and the data observed up to the change:
\begin{equation}
    \mathcal{J}_i[T] \triangleq \sup_{\nu \in \mathbb{N}} \text{esssup}\,\mathbb{E}_{\nu,i}[T - \nu \mid \mathcal{F}_\nu, T > \nu].
\end{equation}

Our goal is to solve the following stochastic optimization problem
$$
  \inf_{(T,D) \in \mathcal{C}(\alpha,\beta)} \mathcal{J}_i[T] 
$$
uniformly over each $i \in \mathcal{I}$. Since the densities are unknown, we cannot implement the min-CUSUM algorithm studied in \cite{warner2024worst}. In this paper, we take a score-based approach and solve the above problem approximately.

\section{Hyvärinen score and Score Based Change Detection}\label{sec:Hyvärinen_score}
In this section, we review the Hyvärinen score and the score-based CUSUM (SCUSUM) algorithm, designed to overcome the limitations associated with the likelihood ratio-based CUSUM when dealing with unnormalized models. The SCUSUM algorithm is proposed and studied in \cite{Wuetal_IT_2024} for QCD in single-stream data. Similar to the conventional CUSUM algorithm, the SCUSUM algorithm is implemented recursively, making it well-suited for online applications due to its low computational and memory requirements.

\subsection{Hyvärinen Score}

Consider a random variable $X$ taking values in $\mathbb{R}^m$. Let $\mathcal{P}$ be a family of probability distributions defined on $\mathbb{R}^m$. Within this family, let $P \in \mathcal{P}$ be the true data-generating distribution of $X$, and let $Q \in \mathcal{P}$ be any alternative, postulated distribution.

First, we introduce proper scoring rules \cite{gneiting2007strictly}. A scoring rule is a function $S: (X, Q) \mapsto \mathbb{R}$ that measures how proper a distribution $Q$ is, in modeling the data $X$. We call the scoring rule  $S(X, Q)$
proper if, $\forall P \in \mathcal{P}$, the expected score $\mathbb{E}_{X \sim P}[S(X, Q)]$ is minimized when $Q = P$, with the minimum taken over all $Q \in \mathcal{P}$. Additionally, $S$ is strictly proper with respect to $\mathcal{P}$ if,  for any $ Q \in \mathcal{P}$, $Q \neq P$, there is 
$$
\mathbb{E}_{X \sim P}[S(X, Q)] > \mathbb{E}_{X \sim P}[S(X, P)].
$$

\begin{definition}
\label{def:Fisher_Divergence}
(\textbf{Fisher Divergence}) The Fisher Divergence from distribution $P$ to $Q$ is defined as 
\begin{equation}
D_F(P \parallel Q) \triangleq  \mathbb{E}_{X \sim P}  \left[ \frac{1}{2}  \| \nabla_X \log p(X) - \nabla_X \log q(X) \|_2^2 \right],    \label{eq:DF_def}
\end{equation} 
where $\|\cdot\|_2$ denotes the Euclidean norm, $\nabla_x$ is the gradient operator, and  $p(X)$ and $q(X)$ are the density functions of distribution $P$ and $Q$, respectively.
\end{definition}
\textit{Remark: } $\nabla_X \log p(X)$ and $\nabla_X \log q(X)$ remain invariant if $p(X)$ and $q(X)$ are scaled by any positive constant with respect to $X$. Thus, the Fisher divergence remains scale-variant under arbitrary constant scaling of density functions.

\begin{definition}
\label{def:Hyvärinen Score}
(\textbf{Hyvärinen Score})
The Hyvärinen score \cite{vincent2011connection} is a function $S_H(X, Q): (X, Q) \mapsto \mathbb{R}$, defined as 
$$
S_H(X, Q) \triangleq \frac{1}{2} \| \nabla_X \log q(X) \|_2^2 + \Delta_X \log q(X),
$$
where $\nabla_X$ is the gradient operator  and $\Delta_X = \sum_{i=1}^m \frac{\partial^2}{\partial x_i^2}$ is the Laplacian operators,  acting on $X = (x_1, \ldots, x_m)^\top$.
\end{definition}

\textit{Remark: } The Hyvärinen score $S_H$ possesses a key property: it is scale-invariant, 
inherited from the scale-invariant nature of the Fisher divergence. 
This property eliminates the need to compute the normalizing constant in unnormalized models. 
Specifically, suppose $\tilde{q}(x)$ is an unnormalized density function, and $q(x)$ is the corresponding 
normalized density defined as
$$
q(x) \;=\; \frac{\tilde{q}(x)}{\int_{x \in \mathcal{X}} \tilde{q}(x) \, dx},
$$
where the normalizing constant $\displaystyle \int_{x \in \mathcal{X}} \tilde{q}(x)\,dx$ is often 
analytically intractable. 
The Hyvärinen score depends only on $\tilde{q}(x)$ and its derivatives, bypassing the need to compute 
the normalizing constant explicitly. 
This makes it particularly useful in scenarios where only $\tilde{q}(x)$ is known, such as in 
Bayesian inference or energy-based models \cite{hyvarinen2005estimation}. Furthermore, it is evident that the Hyvärinen score is strictly proper, because  $D_F(P \parallel Q) > 0$ if $Q \neq P$.

\begin{assumption}[Hyvärinen’s Regularity Conditions \cite{hyvarinen2005estimation}]
\label{assump:Hyvärinen}
The following conditions hold for densities  $p(X) $ and $q(X)$ :  
\begin{enumerate}[label=(\roman*)]
    \item \label{assump:diff_p} The density $ p(X)$ is twice continuously differentiable.
    \item \label{assump:diff_q} The model score function $\nabla_X \log q(X) $ is differentiable.
    \item \label{assump:moments} The expectations
    $ \mathbb{E}_{X \sim p}\left[\|\nabla_X \log p(X)\|_2^2\right] $  
    and  
    $\mathbb{E}_{X \sim p}\left[\|\nabla_X \log q(X)\|_2^2\right] $
    are finite.
    \item \label{assump:boundary} The boundary condition 
    $p(X)\nabla_X \log q(X) \to 0$ holds as $ \|X\|_2 \to \infty $.
\end{enumerate}
\end{assumption}

Under Assumption~\ref{assump:Hyvärinen}, an estimation technique for learning the gradient log density is proposed in \cite{hyvarinen2005estimation} based on minimizing Fisher divergence \ref{def:Fisher_Divergence}, which can also be written as \cite{Wuetal_IT_2024} 
\begin{equation}
D_F(P \parallel Q) = \mathbb{E}_{X \sim P} \left[ \frac{1}{2} \| \nabla_X \log p(X) \|_2^2 + S_H(X, Q) \right]. \label{eq:DF_score}  
\end{equation}
Since $\frac{1}{2} \| \nabla_X \log p(X) \|_2^2$
is constant with respect to $ Q $, minimizing $ D_F(P \parallel Q) $ is equivalent  to minimizing the expected value of $ S_H(X, Q) $, though minimums will differ.

\subsection{Score-Based Quickest Change Detection}

Let  $\{X_n\}$, with $X_n \in \mathbb{R}^m$, be a sequence of independent random vectors such that at a change point, the law changes from the pre-change density $f$ to a post-change density $g$.  Both densities are known in unnormalized and score-based form.

\begin{definition}
    \label{def:SCUSUM score}
(\textbf{SCUSUM Score}) The instantaneous SCUSUM score function $Y_\lambda(X): X \mapsto \mathbb{R}$ is defined as
\begin{equation}
Y_\lambda(X) \triangleq \lambda  \cdot \left( S_H(X, f) - S_H(X, g) \right),   
\label{eq:SCUSUM_score}
\end{equation}
where $S_H(X, f)$ and $S_H(X, g)$ are the Hyvärinen score functions corresponding to the pre- and post-change distributions, respectively, and $\lambda $ is a pre-selected positive multiplier, s.t $\mathbb{E}_{\infty}\left[e^{\lambda \cdot \left( S_H(X, f) - S_H(X, g) \right)} \right] =1$, for details on the selection of $\lambda$, please refer to Lemma \ref{lem:exist_lambda}.

\end{definition}

The statistic for the score-based CUSUM (SCUSUM) algorithm is defined as follows \cite{Wuetal_IT_2024}:
\begin{equation}
    \begin{aligned}
    Z(0) & = 0,  \\
    Z(n) & \triangleq  \left(0, Z(n-1) + Y_\lambda(X_n)\right)^+, \quad \forall n \geq 1,  
    \end{aligned}   \label{eq:SCUSUM_statistics}
\end{equation}
with the stopping time 
\begin{equation}
 T_{\text{SCUSUM}} \triangleq \inf \{n \geq 1 : Z(n) \geq b\},   \label{eq:T_SCUSUM} 
\end{equation}
where $b > 0$ is a threshold, typically pre-selected to control false alarms.

Equivalently, we can rewrite $T_{\text{SCUSUM}}$ defined in (\ref{eq:SCUSUM_statistics}) and (\ref{eq:T_SCUSUM}) as 
\begin{equation}
T_{\text{SCUSUM}} \triangleq \inf \left\{ n \geq 1 : \max_{1 \leq k \leq n} \sum_{t=k}^{n} Y_\lambda(X_t) \geq b \right\}.   
\end{equation}

The SCUSUM algorithm is proposed and analyzed in \cite{Wuetal_IT_2024}, where it is shown that similar to the CUSUM algorithm, setting $b = |\log \alpha |$ ensures that
\begin{equation} 
\label{eq:scusum_far}
\mathbb{E}_\infty [T_{\text{SCUSUM}}] \geq \frac{1}{\alpha}, 
\end{equation}
and for this threshold, we get
\begin{equation} 
\label{eq:scusum_add}
\mathcal{J}[T_{\text{SCUSUM}}] \sim \frac{|\log \alpha |}{\lambda D_F(g \parallel f)}, \; \alpha \to 0,
\end{equation}
where $\mathcal{J}[T]$ is the Lorden's delay metric for a stopping time $T$.  here, let  $g(c)$  and  $h(c)$  be two functions. We write
\[
g(c) \sim h(c) \quad \text{as } c \to c_0, 
\]
to indicate that $ g(c) = h(c)(1 + o(1))$  as  $c \to c_0$.

\section{min-SCUSUM Stopping and Diagnosis Algorithm}
\label{sec:min_SCUSUM_diagnosis}
%



In this section, we discuss the main algorithm of this paper, the min-SCUSUM algorithm. 
In this algorithm, for each channel  $i \in \mathcal{I}$ , we compute the instantaneous SCUSUM score $ Y_i(X_{i,n}) $ using the Hyvärinen scores of the pre- and post-change distributions:  
$$
Y_{\lambda_i}(X_{i,n}) \triangleq \lambda_i \cdot \left( S_H(X_{i,n},f_i) - S_H(X_{i,n},g_i) \right),  
$$  
where  $\lambda_i > 0$ is a pre-selected multiplier satisfying  
$$
\mathbb{E}_{\infty}\left[ e^{ \lambda_i \cdot \left(S_H(X_{i,n}, f_i) - S_H(X_{i,n},g_i) \right)} \right] = 1 
$$ 
(see Lemma~\ref{lem:exist_lambda}).  
For simplicity, we abbreviate $Y_{\lambda_i}(X_{i,n})$ as $Y_i(X_{i,n})$ when no confusion arises.

The SCUSUM statistic  $Z_i(n)$  for channel  $i$  evolves as:  
\begin{equation}
\begin{aligned}
Z_i(0) & = 0,  \\
Z_i(n) & \triangleq \left(0, Z_i(n-1) + Y_{i}(X_{i,n})\right)^{+}.  
\end{aligned}   
\label{eq:SCUSUM_recursion}
\end{equation}  
Equivalently,  $Z_i(n)$  can be expressed in its non-recursive form, which accumulates the maximum cumulative score over all possible change points $1 \leq k \leq n $, as follows:
$$
Z_i(n) = \max_{1 \leq k \leq n} \sum_{t=k}^n Y_i(X_{i,t}).
$$  
\noindent The channel-wise stopping time  $T_i(b)$  is defined as the first instance the statistic  $Z_i(n)$  exceeds a non-negative threshold   $b_i$, 
\begin{equation}
T_i(b) \triangleq \inf \left\{ n \geq 1 : Z_i(n) \geq b_i \right\}.  
\label{eq:sigma_i(b)}
\end{equation}  
For simplicity, we assume all channels share a common threshold  $b \geq 0$, i.e., $b_i =b, \forall i \in \mathcal{I}$. Then, the min-SCUSUM stopping rule declares a change when any channel’s statistic crosses a non-negative threshold  $ b$ :  
\begin{equation}
T(b) \triangleq \min_{i \in \mathcal{I}} T_i(b).  
\label{eq:sigma_b}
\end{equation}  

\noindent $T(b) $ represents the earliest time at which any channel’s SCUSUM statistic  $Z_i(n)$  surpasses  $b$, hence, triggering an alarm.  


At the stopping time $T(b)$ , the diagnosis rule identifies the affected channel using a simple maximum criterion. The estimated affected channel  $D$ is given by:  
\begin{equation}
D \in \arg\max_{i \in \mathcal{I}} Z_i(T(b)),  
\end{equation}  
where ties are resolved arbitrarily.

This rule selects the channel whose SCUSUM statistic  $Z_i(T(b))$  is largest at the time of stopping. By construction, the diagnosed channel corresponds to the first statistic exceeding the threshold  $b$, ensuring alignment between detection and identification.  

The tuple  $(T(b), D) $ constitutes a sequential change diagnosis procedure. In the rest of this paper, we analyze the performance of this procedure, which helps us to
choose the threshold $b$ to satisfy given constraints on the rate of false alarms and the rate of mis-identifications. We also provide the average detection delay performance of this procedure.

\section{Analysis of the min-SCUSUM Algorithm} \label{sec:FA_Delay_properties}
In this section, we analyze the min-SCUSUM algorithm and provide guarantees on its consistency, false alarm rate, delay, and the rate of false isolation.

\subsection{Consistency}

We define
$$
U_i(X_{i,n}) \triangleq S_H(X_{i,n},f_i) - S_H(X_{i,n},g_i),
$$
and abbreviate $U_i(X_{i,n})$ as $U_i(n)$ when it is clear from the context.

For any channel $ i \in \mathcal{I}$, we assume the Fisher divergence from $g_i$ to $f_i$ is finite.  We define the $I_i$ as the Fisher Divergence from $g_i$ to $f_i$, by (\ref{eq:DF_def}) and (\ref{eq:DF_score}), when $ X \sim g_i$:  
\begin{equation}
\begin{aligned}
I_i  & \triangleq D_F(g_i || f_i) = \mathbb{E}_i\left[ \frac{1}{2} \| \nabla_x \log g_i(X) \|_2^2 + S_H(X, f_i) \right] \\
& = \mathbb{E}_i \left[ \frac{1}{2} \| \nabla_x \log g_i(X) \|_2^2 + S_H(X, g_i) \right. \\
& \qquad \left. - S_H(X, g_i) + S_H(X, f_i) \right] \\
& = \mathbb{E}_i\left[ U_i(X_{i,n}) \right].
\end{aligned}  
\label{eq:D_F_P_infty_P_i}
\end{equation}
Also, 
\begin{equation}
\begin{aligned}
D_F(f_i || g_i)  & = \mathbb{E}_\infty\left[ \frac{1}{2} \| \nabla_x \log f_i(X) \|_2^2 + S_H(X, g_i) \right] \\
& = \mathbb{E}_\infty \left[ \frac{1}{2} \| \nabla_x \log f_i(X) \|_2^2 + S_H(X, f_i) \right. \\
& \qquad \left. - S_H(X, f_i) + S_H(X, g_i) \right] \\
& = -\mathbb{E}_\infty\left[ U_i(X_{i,n}) \right].
\end{aligned}
\label{eq:D_F_P_infty_P_i}
\end{equation}

Thus, the min-SCUSUM algorithm defined in (\ref{eq:sigma_i(b)}) and (\ref{eq:sigma_b}) has a negative drift before the change and a positive drift after the change. This shows that the algorithm can detect changes consistently.

\subsection{False Alarm Analysis}

In this section, we obtain a lower bound on the mean time to a false alarm for the min-SCUSUM algorithm. The bound we obtain is identical to the one obtained for the min-CUSUM algorithm in \cite{warner2024worst}. However, our proof technique is different since our algorithm is based on Hyvärinen scores and not likelihood ratios.

\begin{lemma}
\label{lem:exist_lambda} (Existence of positive $\lambda_i$)
Define  $h(\lambda_i) \triangleq \mathbb{E}_{\infty}\left[ e^{\lambda_i \cdot \left( S_H(X_{i,n}, f_i) - S_H(X_{i,n}, g_i) \right)} \right]-1 $, we have either case (a) or case (b).
\begin{enumerate}
\item[(a)]  If $P_{\infty}\left\{ S_H(X_{i,n}, f_i) - S_H(X_{i,n}, g_i)\leq 0\right\} <1 $, then $\exists \lambda_i \in (0,  \infty ),$ s.t $\mathbb{E}_{\infty}\left[e^{\lambda_i \cdot \left( S_H(X_{i,n}, f_i) - S_H(X_{i,n}, g_i) \right)} \right] =1$
\item[(b)] If $P_{\infty}\left\{ S_H(X_{i,n}, f_i) - S_H(X_{i,n}, g_i) \leq 0\ \right\}  =1$, then $\forall \lambda_i \in (0,  \infty )$, $\mathbb{E}_{\infty}\left[e^{\lambda_i \cdot \left( S_H(X_{i,n}, f_i) - S_H(X_{i,n}, g_i) \right)} \right] <1 $
\end{enumerate}
\end{lemma}

\begin{proof}
 The proof can be found in \cite[ Lemma 2]{Wuetal_IT_2024}.
\end{proof}

\begin{theorem} \label{thm:FA}
For the min-SCUSUM stopping time $T(b)$ defined as in (\ref{eq:sigma_i(b)}) and (\ref{eq:sigma_b}), we have
\begin{equation}
    \mathbb{E}_{\infty} [T(b)] \geq \frac{e^b}{  |\mathcal{I}|}.  
\end{equation}
\end{theorem}

\begin{proof}

We first prove some preliminary facts. 
For fixed $ i\in \mathcal{I}, i \neq j$, we define $G_i(n) \triangleq \prod\limits_{k=1}^n \frac{e^{Y_i(X_{i,k})} }{ \mathbb{E}_{\infty} \left[ e^{Y_i(X_{i,k})  } \right]  }$, and note that $ \{ G_i(n) \} $ is a  $P_{\infty}$-martingale with mean 1. This follows because 
\begin{align*}
\mathbb{E}_{\infty} [G_i(n+1) \mid \mathcal{F}_{n} ] &=   \mathbb{E}_{\infty} \left[G_i(n) \cdot \frac{e^{Y_i(X_{i,n+1})} }{ \mathbb{E}_{\infty} \left[ e^{Y_i(X_{i,n+1})  } \right]  }  \middle| \mathcal{F}_{n} \right]  \\
&= G_i(n) \cdot \mathbb{E}_{\infty} \left[   \frac{e^{Y_i(X_{i,n+1})} }{ \mathbb{E}_{\infty} \left[ e^{Y_i(X_{i,n+1}) } \right]  }      \right] \\
&= G_i(n),
\end{align*}
and
\begin{align*}
\mathbb{E}_{\infty} [G_i(n) ] &= \frac{\mathbb{E}_{\infty} \left[\prod\limits_{k=1}^n e^{Y_i(X_{i,k})} \right]   }{ \prod\limits_{k=1}^n \mathbb{E}_{\infty} \left[ e^{Y_i(X_{i,k})  } \right] } =1. 
\end{align*}
Next, for a fix $ i\in \mathcal{I}$, with 
$$ A_i(n) \triangleq \sum\limits_{k=1}^n \prod\limits_{t=k}^n e^{Y_i(X_{i,t})} ,$$
$\{A_i(n) -n \}$ is a $P_{\infty}$-martingale with mean 0. This follows because
\begin{align*}
A_i(n+1)& = \sum_{k=1}^{n} \prod_{t=k}^{n+1}  e^{Y_i(X_{i,t})}   + e^{Y_i(X_{i,n+1})}  \\
& = e^{Y_i(X_{i,n+1})} \left( \sum_{k=1}^{n} \prod_{t=k}^{n}  e^{Y_i(X_{i,t})} \right) + e^{Y_i(X_{i, n+1})}  \\
& = e^{Y_i(X_{i, n+1})}  \cdot (A_i(n) +1),    
\end{align*}
and 
\begin{align*}
& \mathbb{E}_{\infty}[A_i(n+1)-(n+1) | \mathcal{F}_n ] \\
 &= \mathbb{E}_{\infty}\left[e^{Y_i(X_{i, n+1})} \cdot \left(A_i(n) +1 \right) -(n+1) \middle|  \mathcal{F}_n \right] \\
&= (A_i(n) +1) \cdot \mathbb{E}_{\infty}\left[ e^{Y_i(X_{i, n+1})} \right] - (n+1)  \\ 
& = A_i(n)- n.   
\end{align*}
In addition, 
\begin{align*}
\mathbb{E}_{\infty}[A_i(n) -n] &=  \sum_{k=1}^{n}  \prod_{t=k}^n  \mathbb{E}_{\infty}\left[  e^{Y_i(X_{i,t})} \right] -n 
 = 0. 
\end{align*}
Above, we used the fact that $\mathbb{E}_{\infty}\left[e^{Y_i(X_{i,n+1})} \right] =1$, i.e., there exists $\lambda>0$ for which this equation holds. See proof in Appendix \ref{lem:exist_lambda}. 

To lower bound $ \mathbb{E}_{\infty} [T(b)]$, we note that 
$$
\begin{aligned}
T(b) &= \inf \left\{n\geq 1: \max_{i \in \mathcal{I}}   Z_i(n) \geq b \right\}     \\
&= \inf \left\{n\geq 1: \max_{i \in \mathcal{I}}     \max_{1 \leq k \leq n}  \prod_{t=k}^n e^{Y_i(X_{i,t})}     \geq e^b \right\}  \\
& \geq \inf \left\{n\geq 1: \sum_{i=1}^{|\mathcal{I}|}  \sum_{k=1}^n  \prod_{t=k}^n e^{Y_i(X_{i,t})} \geq e^b \right\} \\
& =  \inf \left\{n\geq 1: \sum_{i=1}^{|\mathcal{I}|}  A_i(n) \geq e^b \right\} \triangleq T'(b).
\end{aligned}       
$$
From the property of martingale, we know that the linear combination  $\sum\limits_i (A_i(n) - n)$ is also a $P_{\infty}$-martingale with mean 0. Using the optional sampling theorem, we have 
$$
\begin{aligned}
0& = \mathbb{E}_{\infty} \left[ \sum_{i=1}^{|\mathcal{I}|} (A_i(T'(b)) -T'(b) ) \right]  \\
&=  \mathbb{E}_{\infty} \left[\sum_{i=1}^{|\mathcal{I}|} A_i(T'(b))\right]-|\mathcal{I}| \cdot \mathbb{E}_{\infty} [T'(b)] \\
&\geq  e^b  - |\mathcal{I}| \cdot \mathbb{E}_{\infty} [T'(b)].
\end{aligned}   
$$
This gives
$$\mathbb{E}_{\infty} [T(b)] \geq \mathbb{E}_{\infty} [T'(b)] \geq\frac{e^b}{  |\mathcal{I}|}.  $$    
\end{proof}
Therefore, setting 
$$
b = \log \frac{|\mathcal{I}|}{\alpha}
$$
ensures that $\mathbb{E}_{\infty} [T(b)] \geq \frac{1}{\alpha}$.

\subsection{Delay Analysis}
In this section, we obtain an expression for the asymptotic delay of the min-SCUSUM algorithm.

\begin{theorem} \label{thm:Delay_Analysis}
The delay of $T(b)$ is bounded by
\begin{equation}
 \mathcal{J}_i[T(b)]  = \mathbb{E}_i[T(b)] \leq   \mathbb{E}_i[T_i(b)]  \sim \frac{b}{ \lambda_i I_i }, \quad b \to \infty.   
\end{equation}
When $b= \log \frac{|\mathcal{I}|}{\alpha}$ to satisfy the constraint on the mean time to a false alarm $\frac{1}{\alpha}$, then
\begin{equation}
 \mathcal{J}_i[T(b)] \sim \frac{\log  (|\mathcal{I} |/\alpha) }{ \lambda_i I_i},  \; \alpha \to 0. 
\end{equation}
\end{theorem}

\begin{proof}
The proof is skipped as it follows classical arguments and depends on the delay of the SCUSUM algorithm \eqref{eq:scusum_add} derived in \cite{Wuetal_IT_2024}. 



 
\end{proof}

\subsection{Probability of False Isolation}
In this section, we obtain an upper bound on the probability of mis-identification for the min-SCUSUM algorithm. A similar bound for the min-CUSUM algorithm was obtained in \cite{warner2024worst}. The proof given in \cite{warner2024worst} has several intermediate steps. Many of them are also applicable to our algorithm, as both the min-CUSUM and the min-SCUSUM have a Markovian structure. However, there are certain lemmas in  \cite{warner2024worst} which must be adjusted to account for the fact that our algorithm is based on scores. Consequently, our bound is similar but not identical to the one obtained in \cite{warner2024worst}.

\begin{theorem}
\label{thm:prob_misisolation}
Fix $i, j \in \mathcal{I}, i \neq j, b>0$, and let $\nu \geq 0$. Let $\lambda_i$ be a positive multiplier satisfying  
$$
\mathbb{E}_{\infty}\left[ e^{\lambda_i \cdot \left( S_H(X_{i,n}, f_i) - S_H(X_{i,n}, g_i) \right)} \right] = 1.
$$
The probability of  mis-identifying channel $i$ when the true affected  channel is $j$, $i \neq j$,  given there is no false alarm, is bounded by 
\begin{equation}
    \begin{aligned}
P_{\nu, j} \left(   D=i \mid T(b) >\nu   \right)& \leq
P_{\nu, j}\left(   T_i(b) \leq T_j(b)  \mid T(b)>\nu\right) \\
&\leq  e^{- b} (1+b) \left(1 +\frac{1}{\lambda_{j} I_j } + \zeta_{ij}(b) \right),    
    \end{aligned}
\end{equation}
where  $\zeta_{ij}(b)$ 
is a function of $b$ with $\lim\limits_{b \rightarrow \infty} \zeta_{ij}(b) =0$.
\end{theorem}
  
\begin{proof}
    We provide the proof in Appendix \ref{proof:prob_upperbound}.
\end{proof}

Therefore, the probability of mis-identification is upper bounded by
\begin{equation}
\begin{aligned}
 & \max_{j \in \mathcal{I}} \sup_{\nu \in \mathbb{N}} P_{\nu,j} \left(D \neq j \mid T(b) > \nu \right) \\
 &= \max_{j \in \mathcal{I}}   \sup_{\nu \in \mathbb{N}} \sum_{i \neq j} P_{\nu, j}\left(   D=i \mid T(b)>\nu\right) \\
 & \leq \max_{j \in \mathcal{I}} \sum_{i \neq j}  e^{- b} (1+b) \left(1 +\frac{1}{\lambda_{j} I_j } + \zeta_{ij}(b) \right).   
\end{aligned} \label{eq:final_upperbound}
\end{equation}


Recall that the threshold \( b \) is selected to satisfy two constraints: $\mathbb{E}_\infty[T(b)] \geq 1/\alpha$, and 
$$
\max\limits_{i \in \mathcal{I}} \sup\limits_{\nu \geq 0} P_{\nu,i}(D \neq i \mid T(b) > \nu) \leq \beta. 
$$
%
From the false alarm analysis, it is clear that to ensure the false alarm rate of $\alpha$, the threshold  $b$  must satisfy    
$$  
b \geq \log \frac{|\mathcal{I}|}{\alpha}.
$$ 
To ensure the probability of false identification is below \( \beta \), the threshold \( b \) must satisfy
\[
\max_{j \in \mathcal{I}} \sum_{i \neq j} e^{-b}(1 + b) \left(1 + \frac{1}{\lambda_j I_j} + \zeta_{ij}(b) \right) \leq \beta.
\]
With $b = \log \frac{|\mathcal{I}|}{\alpha}$, by equation~\eqref{eq:final_upperbound}, we have  $b \to \infty$, as  $\alpha \to 0$. Also, the false identification probability is bounded by
\begin{align*}
&\max_{j \in \mathcal{I}} \sup_{\nu \in \mathbb{N}} P_{\nu,j}(D \neq j \mid T(b) > \nu)  \\
&\leq \max_{j \in \mathcal{I}} \sum_{i \neq j} e^{-b}(1 + b) \left(1 + \frac{1}{\lambda_j I_j} + \zeta_{ij}(b) \right) \\
&\leq \max_{j \in \mathcal{I}} \; C \cdot \frac{\alpha}{|\mathcal{I}|} \left(1 + \log \frac{|\mathcal{I}|}{\alpha} \right) \longrightarrow 0, \quad \text{as } \alpha \to 0,
\end{align*}
where \( C > 0 \) is a constant.
This demonstrates that as \( \alpha \to 0 \), this choice of \( b \) drives the mis-identification probability to zero. Therefore, choosing \( b = \log \frac{|\mathcal{I}|}{\alpha} \) provides a unified threshold that simultaneously controls both false alarms and false identifications.

\section{Numerical Results}
\label{sec:numerical}
In this section, we present numerical experiments that validate the theoretical properties of the proposed min-SCUSUM algorithm. We consider two complementary settings: (i) a synthetic multi-stream scenario using high-dimensional data generated by a Gauss-Bernoulli Restricted Boltzmann Machine (GB-RBM), and (ii) a real-world scenario involving video frames captured from a public area in Dublin, Ireland, from the website EarchCam.com \cite{earthcam2024}. 

\subsection{Validating the Upper Bound on the Probability of Misidentification Using Synthesized Data}
This experiment aims to numerically validate the theoretical upper bound on the probability of fault misidentification derived in Theorem~\ref{thm:prob_misisolation}. We simulate a controlled multi-stream environment where only one channel undergoes a distributional change, while the others remain the same.

We present the simulation results for the single-fault multichannel problem described in Section~\ref{sec:problem_formulation}, where there are three candidate channels, and the real change occurs in channel 1, i.e.,  $I = \{1,2, 3\}$ and $j=1$. All the data streams are sampled from a Gauss-Bernoulli Restricted Boltzmann Machine (GB-RBM). The Restricted Boltzmann Machine (RBM) \cite{liao2022gaussian, lecun2006tutorial} is a generative graphical model based on a bipartite structure comprising a layer of "hidden variables" and a layer of "visible variables". 

\begin{figure}[h]
    \centering
    \includegraphics[scale=0.44]{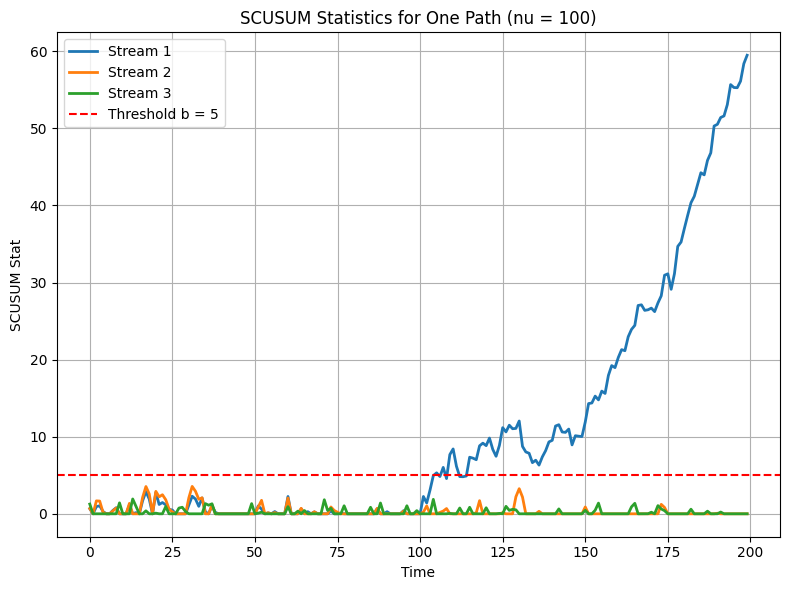}
    \caption{The evolution of the SCUSUM statistics for the three streams.}
    \label{fig:SCUSUM_Stat}
\end{figure}

The GB-RBM consists of:  
\begin{itemize}
    \item Binary-valued hidden variables:    $ H = (h_1, \dots, h_{d_h})^T \in \{0, 1\}^{d_h}$,
    \item Real-valued visible variables:   $  V = (v_1, \dots, v_{d_x})^T \in \mathbb{R}^{d_x}$, and 
    \item Weight matrix representing interactions between hidden and visible units:    $W = (W_{ij}) \in \mathbb{R}^{d_x \times d_h}$.
\end{itemize}
In our numerical experiment, we initialize the GB-RBM parameters as follows:
\begin{itemize}
\item Weight matrix:  
  \[
  W_0 \in \mathbb{R}^{10 \times 5}, \quad \text{with each entry is i.i.d. } \sim \mathcal{N}(0,1). 
  \]
    \item Bias vectors for visible and hidden layers:
    \[
    v_0 \in \mathbb{R}^{10}, \quad \text{with each entry is i.i.d. } \sim \mathcal{N}(0,1), 
    \]
    \[
    h_0 \in \mathbb{R}^{5}, \quad \text{with each entry is i.i.d. } \sim \mathcal{N}(0,1). 
    \]
\end{itemize}
After initializing these parameters, we apply Gibbs sampling to generate independent and identically distributed (i.i.d.) data.

The pre-change data stream is generated from the distribution: 
\begin{align*}
  & f_1 \sim \text{RBM}(W, v, h)\\
  & f_2 \sim \text{RBM}(W+0.2, v, h)\\
  & f_3 \sim \text{RBM}(W+0.1, v, h),
\end{align*}
where the parameters are set as \( W = W_0 \), \( v = v_0 \), and \( h = h_0 \), all sampled independently from a standard normal distribution.
To model the post-change distributions, we introduce structured perturbations to the weight matrix while keeping \( v \) and \( h \) unchanged:
\begin{align*}
& g_1 \sim \text{RBM}(W - 0.1, v, h), \\
& g_2 \sim \text{RBM}(W + 0.1, v, h), \\
&  g_3 \sim \text{RBM}(W + 0.2, v, h).  
\end{align*}
These perturbations induce controlled shifts in the data distributions while preserving their high-dimensional structure.

For change points  $\nu = \{ 0, 20, 100\} $, we simulate 10,000 independent sample paths of the SCUSUM statistics. 
To illustrate the behavior of the SCUSUM statistics over time, Figure~\ref{fig:SCUSUM_Stat} shows one sample path of the SCUSUM statistics for the three monitored channels, with stream 1 as the truly affected stream. As expected, its SCUSUM statistic exhibits a clear upward trend after the change point, while the statistics for the unaffected channels (Monitors 2 and 3) remain near zero. If we control the false alarm rate to be $\alpha = 0.02$, we then set $b = \log (d/ \alpha) = 5$, with expected delay \(  = \frac{\log (d /\alpha)}{ \lambda_1 D_F(g_1||f_1 )} = 6.4\) in stream 1. 
\begin{figure}
    \centering
    \includegraphics[scale=0.42]{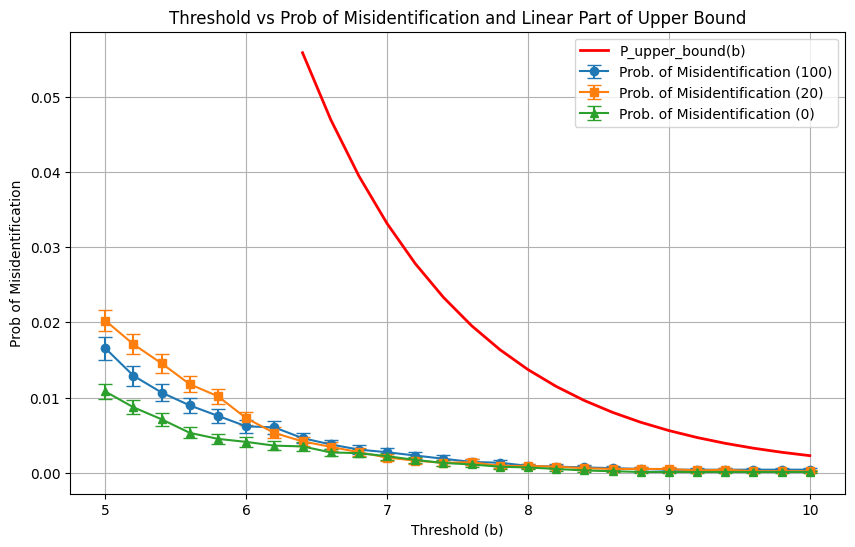}
    \caption{The plots of the probability of misidentification for three possible change points $\nu=0,20,100$ as a function of the stopping threshold $b$. }
    \label{fig:upper_bound}
\end{figure}
 Figure~\ref{fig:upper_bound} shows that for each change point  $\nu \in \{0, 20, 100\}$, the conditional probability of fault misidentification decreases as the decision threshold $b$  increases. This behavior aligns with the theoretical expectation that a higher threshold reduces the likelihood of selecting an incorrect stream. The red curve in the figure corresponds to the first-order upper bound approximation  $C(1 + b)e^{-b} $, highlighting that the empirical misidentification probabilities consistently fall below this analytical bound.


\subsection{Detection on Real-World Video Streams}

To further demonstrate the practical applicability of our method, we apply the min-SCUSUM algorithm to real video data captured from the temple bar area of Dublin, Ireland. These videos are obtained from publicly available live streams through EarthCam.com \cite{earthcam2024}. The goal is to identify abrupt changes in visual activity that correspond to unusual events.

\begin{figure}[h]
    \centering
    \subfloat[Pre 1]{\includegraphics[width=0.1\textwidth]{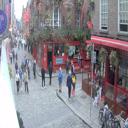}}
    \hspace{0mm}
    \subfloat[Pre 2]{\includegraphics[width=0.1\textwidth]{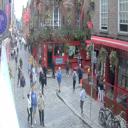}}
    \hspace{0mm}
    \subfloat[Pre 3]{\includegraphics[width=0.1\textwidth]{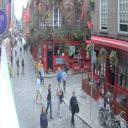}}
    \hspace{0mm}
    \subfloat[Pre 4]{\includegraphics[width=0.1\textwidth]{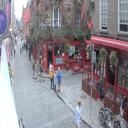}}
    \hspace{2mm}
    \subfloat[Post 1]{\includegraphics[width=0.1\textwidth]{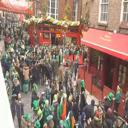}}
    \hspace{0mm}
    \subfloat[Post 2]{\includegraphics[width=0.1\textwidth]{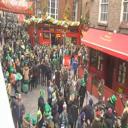}}
    \hspace{0mm}
    \subfloat[Post 3]{\includegraphics[width=0.1\textwidth]{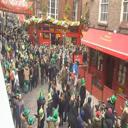}}
    \hspace{0mm}
    \subfloat[Post 4]{\includegraphics[width=0.1\textwidth]{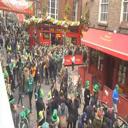}}
    \caption{Sample video frames before and after change.}
    \label{fig:video_frames_excerpts}
\end{figure}

\subsubsection{Data Preprocessing} 
We collect video segments from multiple dates and events, including normal weekdays, busy pub nights, and St. Patrick’s Day celebrations. Representative frames are extracted at regular intervals and resized to \(64 \times 64\) pixels to facilitate neural network training. Figure~\ref{fig:video_frames_excerpts} presents a sample excerpt of video frames before and after a distributional shift. Frames (a)–(d) correspond to the pre-change period, capturing typical daytime pedestrian activity, while frames (e)–(h) are taken from the post-change period, reflecting a significantly altered scene during a public event. This visual comparison illustrates the nature of the distributional change targeted by the detection algorithm.

\subsubsection{Modeling Setup} 
We treat each image frame as a high-dimensional observation from a stream. 
We create three artificial streams. 
Frames from normal daylight hours serve as the pre-change distribution \( f=f_i \) for each stream, while different event types are modeled as distinct post-change distributions for each of the three streams:
\begin{itemize}
    \item Music pub night frames are used to generate post-change samples for Stream 1 ($g_1$). 
    \item St. Patrick’s Day frames are used to generate post-change samples for Stream 2 ($g_2$).
    \item Altercation or crowd disturbance frames are used to generate post-change samples for Stream 3 ($g_3$).
\end{itemize}
We train a U-Net-based score network \cite{ronneberger2015u, song2019generative} to estimate the Hyvärinen scores of these distributions using images from each setting. We conducted three different experiments, with the change affecting a single stream in each case. For example, in the experiment where Stream 1 is affected, its post-change law is switched to $g_1$ at the time of change, while the laws of the other streams are maintained at $f$. 

Below, we only report results for the experiment where the change occurs in Stream 1. The results for the other experiments are similar and are not reported in this paper.

\subsubsection{Results} 
For the chosen experiment, we compute the SCUSUM statistics over time for each stream. 
Figure~\ref{fig:z_k_3_streams} shows the SCUSUM statistics computed across three streams when the change occurs in channel 1 at time index 1000. The first half of the data corresponds to normal daytime activity, while the second half represents the post-change scenario—here, a music pub night. As expected, the SCUSUM statistic for the affected stream (SCUSUM 1) exhibits a sharp upward trend after the change point, while the statistics for the unaffected streams (SCUSUM 2 and SCUSUM 3) remain near zero throughout. Figure~\ref{fig:drift_behavior} plots the instantaneous SCUSUM increments, given by the difference of Hyvärinen scores \( S_H(X_k, f_i) - S_H(X_k, g_i) \), for each stream. The affected stream (channel 1) shows a clear shift from negative to positive drift at the change point (indexed at 1000), consistent with the expected behavior under the min-SCUSUM algorithm. The unaffected streams maintain steady negative drift throughout, further validating the robustness of the method in isolating the true change.

\begin{figure}
    \centering
    \includegraphics[width=0.85\linewidth] {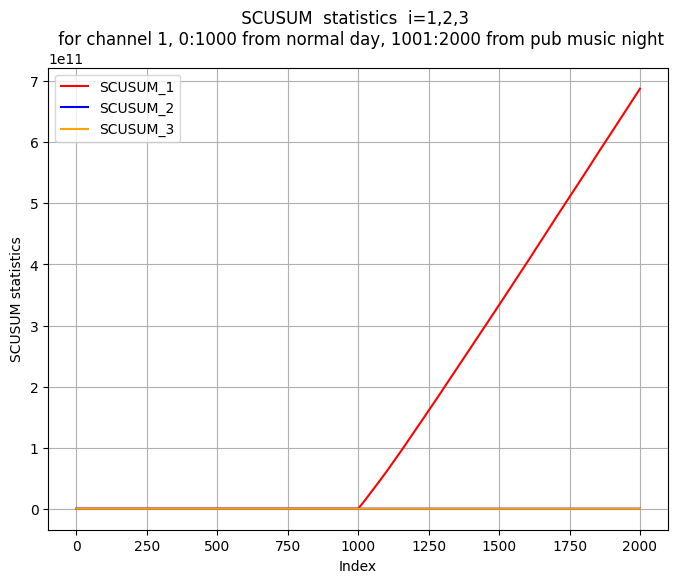}
    \caption{SCUSUM statistics across three streams.}
    \label{fig:z_k_3_streams}
\end{figure}

\begin{figure}
    \centering
    \includegraphics[width=0.85\linewidth]{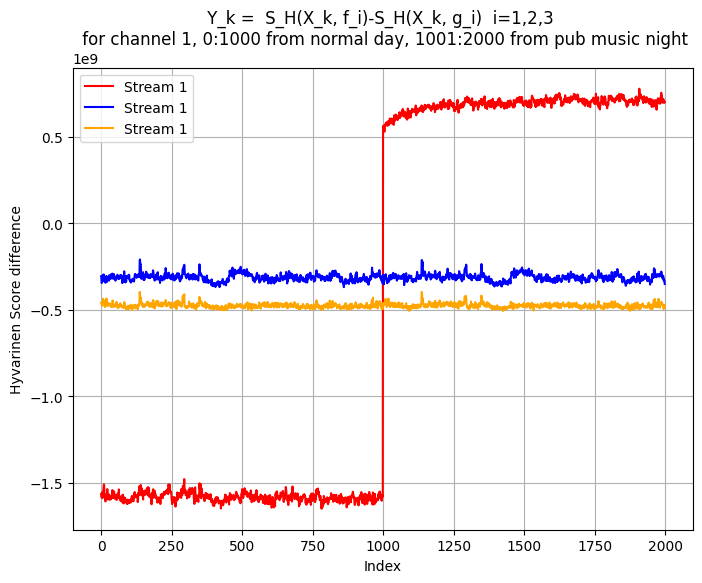}
    \caption{Negative drift before change, positive drift after.}
    \label{fig:drift_behavior}
\end{figure}


\appendices



\section{Some Properties of the Min-SCUSUM Stopping Rule}
\label{sec:properties}

In this appendix, we discuss some notation and properties of the min-SCUSUM algorithm. 
Recall from \eqref{eq:sigma_i(b)} that $T_i(b)$ is the stopping time of the $i$-th channel with threshold $b$. Suppose the change occurs at $\nu= 0$ in the stream $j$. Then the drift of the $i$th stream for $j \neq i$ will still be negative after the change. Consequently, using \eqref{eq:scusum_far} or Thereom~\ref{thm:FA}, we can lower bound the stopping time of the $i$-th channel $T_i(b)$ under measure $P_j$ as 
\begin{equation}
     \mathbb{E}_j[T_i(b)] \geq e^{ b}.  \label{eq:E_j_sigma_i}
\end{equation}



To bound $\mathbb{E}_j[T_j(b)]$, we introduce several auxiliary quantities.
For the $j$-th channel, we define  the SCUSUM random walk as, 
\begin{equation}
  V_j(n) \triangleq \sum_{k=1}^n Y_j(X_{j,k}) .   
\end{equation}
We define overshoot of the random walk $R_j$ as the function of  stopping threshold $b$ and stopping time $T_j(b)$:
\begin{equation}
R_j(T_j(b)) \triangleq V_j(T_j(b)) - b .  \label{eq:remainder}
\end{equation}
From \cite{lorden1970excess}, we get
\begin{align}
& \sup_{b \geq 0} \mathbb{E}_j[ R_j(T_j(b)) ] 
\leq \frac{\mathbb{E}_j\left[ Y_j (X_{j,n})^+ \right]^2}{\mathbb{E}_j[ Y_j(X_{j,n}) ]} \nonumber \\
&= \frac{
\mathbb{E}_j \left[ \lambda_j \cdot \left( S_H(X_{j,n}, f_j) - S_H(X_{j,n}, g_j) \right)^+ \right]^2
}{
\lambda_j \cdot I_j
}.
\label{eq:lorder_excess_bound}
\end{align}
Substituting (\ref{eq:remainder}) into (\ref{eq:lorder_excess_bound}) and using Wald's identity, we then obtain an upper bound on the stopping time $T_j(b)$:
\begin{align*}
\mathbb{E}_j[ R_j(T_j(b)) ] &= \mathbb{E}_j[Y_j(X_{j,1})] \cdot \mathbb{E}_j[T_j(b)]  \\
& \leq b + \frac{\lambda_j \mathbb{E}_j[ U_j^+ ]^2 }{I_j}.
\end{align*}
This gives us
\begin{equation}
\mathbb{E}_j[T_j(b)] \leq \frac{b}{\lambda_j {I_j}} + \frac{\mathbb{E}_j\left[ ( U_j^+ )^2 \right]}{  I_j^2 } . \label{eq:E_j_sigma_j}
\end{equation}


Next, we establish lower bounds for the scenario where $ Z_i$ starts from a value $x$ within the interval $[0, b]$. We define $L_{ij}(x; b)$ as the expected value of $T_i(b)$ under measure $P_j$, where $i \neq j$,  given that $Z_i $ is initialized at some $x \in [0, b]$, i.e.,
\begin{equation}
L_{ij}(x; b) \triangleq \mathbb{E}_j \left[ T_i(b) \mid Z_i(0) = x \right], \quad x \in [0, b], i \neq j.  \label{eq:def_L_ij}
\end{equation}
If there is no ambiguity, we write $ L_{ij}(0; b) = \mathbb{E}_j[T_i(b)]$.
Clearly, \( L( x ; b) \) is a non-increasing w.r.t $x$, i.e., 
\begin{equation}
   L_{ij}(0; b) \geq   L_{ij}(x; b) . 
   \label{eq:monotone_Lij}
\end{equation}

\begin{lemma}
\label{lem:lowerbound L_ij} 
We have the following bound:
$$
L_{ij}(x; b)  \geq l_{ij}(x; b), \quad  \forall  x \in [0, b],
$$
where 
\begin{equation}
l_{ij}(x; b) \triangleq \frac{x -e^{-(b-x)}(\omega_i + b) }{ \lambda_i \cdot D_F(f_i ||g_i)} + \left( 1-e^{-(b-x)} \right) \cdot L_{ij}(0;b).
\end{equation}

\end{lemma}

\begin{proof}
We have defined  $V_i(n)$ as
$$
V_i(n) \triangleq \sum_{k=1}^n Y_i (X_{i,k}) =  \sum_{k=1}^n \lambda_i [S_H(X_{i,k}, f_i) -S_H(X_{i,k}, g_i) ] . 
$$
We denote  $(\sigma_i(x), d_i(x))$  as the stopping time and corresponding decision rule for distinguishing between  $f_i$  and $ g_i$ , when the detection statistic  $Z_i$  is initialized at  $x \in [0, b]$. We monitor the statistics until it reaches either the lower boundary at 0 or the upper boundary at  $b$. This construction follows the sequential testing framework described in \cite[Chapter 3]{tartakovsky2014sequential}. Equivalently, the procedure can be viewed in terms of the random walk $V_i(n) $ initialized at 0, which is stopped when it exits the interval  $(-x, b - x)$, i.e.,
\[
\begin{aligned}
  \sigma_i(x) & = \inf \{ n \in \mathbb{N} : V_i(n) \notin (-x, b - x) \}, \\
  d_i(x) & =
\begin{cases}
0, & \text{if } V_i(\sigma_i(x)) \leq -x, \\
1, & \text{if } V_i(\sigma_i(x)) \geq b - x.
\end{cases}
\end{aligned}
\]
In both cases, the process terminates upon reaching the upper boundary and restarts upon hitting the lower one.

Define $p(x) \triangleq P_j(d_i(x) = 1). $  By the definition of $L_{ij}(x; b),  x \in [0, b]$, we get
\begin{align}
  &  L_{ij}(x; b) \triangleq \mathbb{E}_j[T_i(b) \mid Z_i(0) = x ] \nonumber\\
   & = p(x)\mathbb{E}_j[ \sigma_i(x) \mid d_i(x) = 1]  \nonumber\\
   & \quad + (1-p(x)) (\mathbb{E}_j[ \sigma_i(x)\mid d_i(x) = 0]  + L_{ij}(0;b) )  \nonumber \\
   & = \mathbb{E}_j[\sigma_i(x)] + (1 - p(x)) \cdot L_{ij}(0; b). \label{eq:L_ij_inparts}
\end{align}    

So, to lower bound $L_{ij}(x; b) $, we just need to lower bound  $1- p(x) $ and   $\mathbb{E}_j[\sigma_i(x)]$, separately. 

We now upper bound $p(x)$. Easy to check  that  $ \{ \exp\left( \sum_{k=1}^n Y_i(X_{i,k}) \right) \}_{n\in \mathbb{N}} $
is a  martingale under $P_j$ with expectation 1.  By the definition of  $p(x) $ and the supermartingale inequality, we get 
\begin{equation}
\begin{aligned}
& p(x) = P_j(V_{i}(\sigma_i(x)) \geq b-x)  \\ 
&\leq P_j \left( \sup_{n \geq 1} e^{V_{i}(n)}  \geq  e^{b-x} \right) \leq e^{-(b-x)}. \label{eq:p_x}
\end{aligned}     
\end{equation}

Next, we derive a lower bound for  $\mathbb{E}_j[\sigma_i(x)]$. Following a similar approach to that in \cite[Section 3.1.2]{tartakovsky2014sequential}, we can modify the bound for $\mathbb{E}_j[V_{i}(\sigma_i(x))]$  to obtain a corresponding bound for  $\mathbb{E}_j[\sigma_i(x)]$. By Wald's identity and (\ref{eq:D_F_P_infty_P_i}), we have
\begin{align*}
\mathbb{E}_j\bigg[ V_{i}(\sigma_i(x)) \bigg] & = \mathbb{E}_j \left[ \sum_{k=1}^{\sigma_i(x)}  \lambda_{i} U_i(X_{i,k}) \right]  \\
&= \mathbb{E}_j[ \sigma_i(x) ] \cdot \mathbb{E}_j[ \lambda_{i} U_i(X_{i,1}) ] \\
& =  - \lambda_i \mathbb{E}_j[ \sigma_i(x) ] \cdot D_F(f_i ||g_i),
\end{align*}
where we used the fact that $\mathbb{E}_j[ U_i(X_{i,1})] = -D_F(f_i ||g_i)$. 
Therefore, we have 
\begin{equation}
 \mathbb{E}_j[ \sigma_i(x) ]  = \frac{ - \mathbb{E}_j[ V_i(\sigma_i(x)) ]  }{  \lambda_{i} \cdot D_F(f_i ||g_i)  }.   \label{eq:E_j_sigma_i_expand}
\end{equation}

We define 
\begin{equation}
 \omega_i \triangleq \sup\limits_{t> 0}\mathbb{E}_\infty[ Y_{i}(X_{i,1}) -t \mid Y_{i}(X_{i,1}) \geq t].   
\end{equation}
The quantity $\omega_{i}$ represents the supremum of the expected overshoot of $Y_{i}(X_{i,1})$ above a threshold $t$ at the first time it crosses that threshold from below. It serves as an upper bound for the expected overshoot of the random walk $V_i$ under the measure $P_j$ above any non-negative threshold $t$. It can be shown that $\omega_{i}$ is finite \cite[Chapter 3]{tartakovsky2014sequential}.
Also, we notice that 
\begin{align*}
 & \mathbb{E}_j[ V_i(\sigma_i(x)) \mid V_i(\sigma_i) \geq b-x]  \quad \forall  x \in [0,b]   \\
 & = \mathbb{E}_j \bigg[ V_i(\sigma_i(x)) -(b-x) \mid V_i(\sigma_i) \geq b-x \bigg] +(b-x)  \\
 & \leq  \sup_{t\geq 0}\mathbb{E}_j[ V_i(\sigma_i(x)) -t \mid V_i(\sigma_i) \geq t] +(b-x) \\
 & = \sup_{t\geq 0}\mathbb{E}_\infty[ V_i(\sigma_i(x)) -t \mid V_i(\sigma_i) \geq t] +(b-x) \\
 & \leq \omega_i +(b-x).
\end{align*}
By conditional expectation, we can write 
\begin{equation}
\begin{aligned}
\mathbb{E}_j[ V_i(\sigma_i(x)) ]   = p(x) \cdot \mathbb{E}_j[ V_i(\sigma_i(x)) \mid V_i(\sigma_i(x)) \geq b-x  ] \\
+ (1-p(x)) \cdot \mathbb{E}_j[ V_i(\sigma_i(x)) \mid V_i(\sigma_i(x)) \leq -x].      
\end{aligned}
\label{eq:E_j_V_i_T}
\end{equation}
Substitute (\ref{eq:p_x}), (\ref{eq:E_j_sigma_i_expand}) and (\ref{eq:E_j_V_i_T}) into (\ref{eq:L_ij_inparts}), we get
$$
\begin{aligned}
  L_{ij}&(x;b)   =   \frac{ -p(x)  \mathbb{E}_j[ V_i(\sigma_i(x)) \mid V_i(\sigma_i(x)) \geq b-x  ]} {  \lambda_i   D_F(f_i ||g_i)  } \\
 & \quad + (1-p(x))  \\
 &\quad \quad \bigg(\frac{ \mathbb{E}_j[ - V_i(\sigma_i(x)) \mid V_i(\sigma_i(x)) \leq -x]   } {  \lambda_i   D_F(f_i ||g_i)  } +L_{ij}(0;b)  \bigg)\\
& \geq  e^{-(b-x)}  \frac{-(\omega_i + b-x)}{ \lambda_i   D_F(f_i ||g_i) }  \\
& \quad + \left( 1-e^{-(b-x)} \right) \left[ \frac{x}{ \lambda_i   D_F(f_i ||g_i) } + L_{ij}(0;b) \right] \\
& = \frac{x -e^{-(b-x)}(\omega_i + b) }{ \lambda_i   D_F(f_i ||g_i)} + \left( 1-e^{-(b-x)} \right)   L_{ij}(0;b) \\
& = l_{ij}(x;b).
\end{aligned}   
$$
\end{proof}

\section{Proof of the main theorem}
\label{proof:prob_upperbound}

We prove our main theorem, Theorem (\ref{thm:prob_misisolation}), as follows.

\begin{proof} 
\label{proof:prob_upperbound} 
The proof of our main theorem follows a structure analogous to that in \cite{warner2024worst}, 
with the key difference that their analysis is based on log-likelihood ratios, 
whereas ours employs differences of Hyvärinen scores. 
In line with their approach, we assume the threshold $b$ is sufficiently large so that 
\[
P_{\nu,j}(T_i > T_j \mid T>\nu) > 0,
\]
the justification of which is provided in Lemma~\ref{lem:strict_positive_conditional_P} below.

Recall that 
$$
P_{\nu, j} \left(   D=i \mid T(b) >\nu   \right) \leq
P_{\nu, j}\left(   T_i(b) \leq T_j(b)  \mid T(b)>\nu\right).
$$
Thus, to obtain a bound on $P_{\nu, j} \left(   D=i \mid T(b) >\nu   \right)$, it is enough to bound $P_{\nu, j}\left(   T_i(b) \leq T_j(b)  \mid T(b)>\nu\right)$. 

By the law of total probability, we get
\begin{equation}
\begin{aligned}
&P_{\nu, j}(T_i \leq T_j \mid T> \nu) = 1 - P_{\nu, j}(T_i > T_j \mid T> \nu)  \\
&= 1 - \frac{\mathbb{E}_{\nu, j}[T_i - T_j \land T_i \mid T> \nu]}{\mathbb{E}_{\nu, j}[T_i - T_j \mid T> \nu, T_i > T_j]}.
\end{aligned}    \label{eq:p_true_isolation}
\end{equation}

Note that $\{T_i > T_j \geq T>\nu\} \in \mathcal{F}_{T_j}$, 
since this event is determined at time $T_j$. 
Thus, the quantity $\mathbb{E}_{\nu, j}[T_i - T_j \mid T> \nu, T_i > T_j, \mathcal{F}_{T_j}]$ represents the additional time for channel $i$  to reach the threshold  $b$, as if it starts from $Z_i(T_j)$.  

Furthermore, using the monotonicity in  initial condition of channel $i$ (see   \eqref{eq:monotone_Lij} in Appendix~\ref{sec:properties} for details), we obtain  
\begin{equation}
\begin{aligned}
& \mathbb{E}_{\nu, j} [ T_i - T_j \mid T> \nu, T_i > T_j , \mathcal{F}_{T_j}] \\
& = \mathbb{E}_{\nu, j} \left[ T_i  \mid Z_i(0)= Z_i(T_j)\right] \cdot \mathbb{I} \left( T>\nu, T_i > T_j \right)   \\ 
& \leq  \mathbb{E}_j[T_i].     
\end{aligned} \label{eq:T_i_T_j}   
\end{equation}
Combine (\ref{eq:T_i_T_j}) and  (\ref{eq:p_true_isolation}), we then get 
\begin{equation}
\begin{aligned}
& P_{\nu, j}(T_i \leq T_j \mid T> \nu)  \le\\
& \frac{\mathbb{E}_j[T_i] - \mathbb{E}_{\nu, j}[T_i - \nu \mid T> \nu] }{\mathbb{E}_j[T_i]} + \frac{\mathbb{E}_{\nu, j}[T_j \land T_i - \nu \mid T> \nu]}{\mathbb{E}_j[T_i]}  .\label{eq:P_nu_j}
\end{aligned}    
\end{equation}

The second term of (\ref{eq:P_nu_j}) can be bounded as
\begin{align}
&\frac{ \mathbb{E}_{\nu, j}[T_j \land T_i - \nu \mid T > \nu] }{\mathbb{E}_j[T_i]}  
   \leq \frac{ \mathbb{E}_{\nu, j}[T_j - \nu \mid T > \nu] }{\mathbb{E}_j[T_i]}
 \nonumber \\
&= \frac{ \mathbb{E}_{\nu, j}[T_j  \mid Z_j(0) = Z_j(\nu) ]}{\mathbb{E}_j[T_i]} 
 \nonumber \\
&\leq \frac{\mathbb{E}_j[T_j \mid Z_j(0) = 0]}{\mathbb{E}_j[T_i]} \label{eq:ratio_a}  \\
&= \frac{\mathbb{E}_j[T_j]}{\mathbb{E}_j[T_i]}\leq e^{- b}\left(\frac{b}{\lambda_{j} I_j} +  \frac{\mathbb{E}_j\!\left[ ( U_j^+ )^2 \right]}{  I_j^2 }\right). \label{eq:ratio_b} 
\end{align}
The equation \eqref{eq:ratio_a} holds because $Z_j$ is stochastically monotone. 
We obtain \eqref{eq:ratio_b} using properties of the min-SCUSUM stopping rule, provided in Appendix~\ref{sec:properties} (see \eqref{eq:E_j_sigma_j} and \eqref{eq:E_j_sigma_i} for details),

We notice that on event $\{ T>\nu\}$,  for every channel $i$ at time $\nu$, the SCUSUM statistics $Z_i(\nu) < b$,  and $T_i-\nu$ is a positive-valued function, depending on $Z_i(\nu)$. Hence, we have 
\begin{align}
 & \mathbb{E}_{\nu, j}[T_i - \nu \mid  T> \nu, \mathcal{F}_\nu] \nonumber\\
&= \mathbb{E}_j \left[ T_i(b) \mid Z_i(0) = Z_i(\nu)\wedge b \right] \cdot \mathbb{I}(\{T> \nu\}) \nonumber \\
& = \mathbb{E}_{\infty} \left[ L_{ij}(Z_i(\nu)\wedge b ; \, b) \mid T> \nu \right].   \label{eq:T_i_nu}
\end{align}
Equation (\ref{eq:T_i_nu}) holds because $Z_i(\nu)$ has the same distribution under  $P_{\nu, j}$ and  $P_{\infty}$. 

The first term of (\ref{eq:P_nu_j}) can be bounded as
\begin{align}
&\frac{\mathbb{E}_j[T_i] - \mathbb{E}_{\nu, j}[T_i - \nu \mid T > \nu]}{\mathbb{E}_j[T_i]}  \nonumber\\
& = \frac{\mathbb{E}_j[T_i] - \mathbb{E}_{\infty} \left[ L_{ij}( Z_i(\nu)\wedge b  ;b) \mid T > \nu \right] }{\mathbb{E}_j[T_i]}  \nonumber \\
&\leq \frac{L_{ij}(0;b) - l_{ij}(0;b)  }{\mathbb{E}_j[T_i]}   \nonumber\\
& \quad + \frac{\mathbb{E}_{\infty} \left[l_{ij}(0;b) - l_{ij}( Z_i(\nu)\wedge b  ;b) \mid T > \nu \right] }{\mathbb{E}_j[T_i]} . \label{eq:first_term}   
\end{align} 
This equation \eqref{eq:first_term} holds because of lemma~\ref{lem:lowerbound L_ij}. 

Setting $x=0$ in Lemma~\ref{lem:lowerbound L_ij} gives the numerator of the first term in \eqref{eq:first_term} as
\begin{equation}
 L_{i j}(0 ; b)-l_{i j}(0 ; b) 
 = e^{-b} \bigg(\frac{  \omega_i +b}{\lambda_i \cdot D_F(f_i ||g_i)} +  L_{ij}(0;b) \bigg). \label{eq:L_ij_l_ij}
\end{equation}
We now upper bound the numerator of the second term in \eqref{eq:first_term}. Since there is no false alarm, $Z_i(\nu)\wedge b = Z_i(\nu)$, and hence
\begin{align}
&  \mathbb{E}_{\infty} \left[l_{i j}(0 ; b)-l_{i j}\left(Z_i(\nu)\wedge b ; b\right) \middle| T>\nu \right] \nonumber\\     
&= \mathbb{E}_{\infty} \left[-\int_0^{Z_i(\nu)} l_{i j}^{\prime}(x ; b)  \cdot \mathbb{I}\left(0 \leq Z_i(\nu) \leq b \right) d x \middle| T>\nu \right]  \label{eq:l_ij_diff_a} \\
&=  \frac{\int \mathbb{I}(T>\nu) \int_0^b -l_{i j}^{\prime}(x ; b) \cdot \mathbb{I}\left(Z_i(\nu) \geq x\right) d x  d  P_{\infty} }{P_{\infty}(T>\nu)} \nonumber \\
& = \int_0^b -l_{i j}^{\prime}(x ; b) P_{\infty}\left(Z_i(\nu) \geq x \mid T>\nu\right) d x  \label{eq:fubini}
\end{align}
\begin{align}
& \leq  \int_0^b \bigg[  e^{-(b-x) }  \bigg(\frac{\omega_i + b}{\lambda_i  D_F(f_i ||g_i)} + L_{ij}(0;b)\bigg)  \nonumber\\
& \qquad \qquad \qquad \qquad  P_{\infty}\left(Z_i(\nu) \geq x \mid T>\nu\right)\bigg] d x 
\label{eq:l_ij_diff_b} \\
& \leq \int_0^b    e^{-(b-x) }  \bigg(\frac{\omega_i + b}{\lambda_i  D_F(f_i ||g_i)} + L_{ij}(0;b)\bigg)  e^{-x} d x \label{eq:l_ij_diff_c} \\
& = b e^{-b}  \bigg(\frac{\omega_i + b}{\lambda_i  D_F(f_i ||g_i)} + L_{ij}(0;b)\bigg) .\label{eq:l_ij_diff_d} 
\end{align}
The equation \eqref{eq:l_ij_diff_a} follows from the fundamental theorem of calculus, where $l_{i j}^{\prime}(x ; b) $ is the derivative of $l_{i j}(x ; b) $. Equation \eqref{eq:fubini} follows from Fubini's theorem and 
\eqref{eq:l_ij_diff_b} comes from Lemma \ref{lem:bound l_ij}. The equation \eqref{eq:l_ij_diff_c} is true because of Lemma \ref{lem:main_condition_holds} below.
Substituting \eqref{eq:L_ij_l_ij} and \eqref{eq:l_ij_diff_d} into \eqref{eq:first_term}, and then combining with  \eqref{eq:ratio_b}, we can upper bound \eqref{eq:P_nu_j} as
$$
\begin{aligned}
& P_{\nu, j}\left(T_i \leq T_j \mid T>\nu\right) \\
& \leq \frac{\mathbb{E}_j\left[T_i\right]-\mathbb{E}_{\infty}\left[L_{i j}\left(Z_i(\nu) \wedge b ; b\right) \mid \sigma>\nu\right]}{\mathbb{E}_j\left[T_i\right]}+\frac{\mathbb{E}_j\left[T_j\right]}{\mathbb{E}_j\left[T_i\right]} \\
& \leq e^{- b} (1+b) \left(1 +\frac{1}{\lambda_{j} I_j } + \zeta_{ij}(b) \right) .
\end{aligned} 
$$
Where  $\zeta_{ij}(b)\triangleq \frac{(\omega_{i}+b)e^{-b} }{ \lambda_{i} D_F(f_i||g_i) } +   \frac{\mathbb{E}_j\left[  U_j^+  \right]^2}{  I_j^2 (1+b)}$,   and $\lim\limits_{b \rightarrow \infty} \zeta_{ij}(b) =0$.
\end{proof}

\section{Proofs of Some lemmas}

\begin{lemma}
\label{lem:main_condition_holds}
Let $Z_i(n)$ be as defined in \eqref{eq:SCUSUM_recursion} and $T(b)$ be the min-SCUSUM algorithm defined in \eqref{eq:sigma_b}. Then
\[
  P_{\infty}\left(Z_i(\nu) \geq x \mid T(b)>\nu\right) \leq e^{-x}, \quad \forall x \geq 0, \forall \nu \geq 0.    
\]
\end{lemma}

\begin{proof}
The statement is true because
\begin{align}
P_{\infty}&\left(Z_i(\nu) \geq x \mid T(b)>\nu \right)  
   \nonumber\\
&= P_{\infty}\left(Z_i(\nu) \geq x \mid T_k(b)>\nu, \forall k \in \mathcal{I} \right)  
   \label{eq:condition_a} \\
&= P_{\infty}\left(Z_i(\nu) \geq x \mid T_i(b)>\nu \right)  
   \label{eq:condition_b} \\
&\leq  P_{\infty}\left(Z_i(\nu) \geq x \right)  
   \label{eq:condition_c} \\
&\leq e^{-x} \cdot \mathbb{E}_{\infty}\!\left[ \exp\!\left( \sum_{k=1}^\nu Y_i(X_{i,k}) \right) \right]  
   \label{eq:condition_d} \\
&= e^{-x}.    \label{eq:condition_e}
\end{align}
Here, \eqref{eq:condition_a} follows from the definition of $T(b)$ and \eqref{eq:condition_b} uses the independence of the individual channels. 
The equation \eqref{eq:condition_c} is true because of \cite[Theorem 1]{pollak1986convergence}, provided we can show that $Z_i$ is stochastically monotone. This is shown below. The equation \eqref{eq:condition_d} follows from the supermartingale inequality, and 
\eqref{eq:condition_e} holds because
$ \{\exp( \sum_{k=1}^n Y_i(X_{i,k})) \}_{n\ge 1} $ is a martingale under  $P_{\infty}$ with expectation 1.
Finally, $Z_i(n)$ is stochastically monotone since 
$P_\infty(Z_i(t) \geq y \mid Z_i(0) = x)$ is non-decreasing in $x$. Indeed, for $x_1 \leq x_2$,
\[
\begin{aligned}
  &P_\infty(Z_i(n) \geq y \mid Z_i(0) = x_1) \\
  &= P_\infty\!\left(\max_{1 \leq k \leq n }\sum_{t=k}^n Y_i(t)  \geq y-x_1 \right) \\   
  &\leq  P_\infty\!\left(\max_{1 \leq k \leq n }\sum_{t=k}^n Y_i(t)  \geq y-x_2 \right) \\ 
  &= P_\infty(Z_i(n) \geq y \mid Z_i(0) = x_2).  
\end{aligned}
\]
Moreover, $Z_i(n)$ is right-continuous in $x$ for every $y \geq 0$.  
For $\epsilon > 0$,
\[
\begin{aligned}
& P_\infty(Z_i(t) \geq y \mid Z_i(0) = x + \epsilon) 
   - P_\infty(Z_i(t) \geq y \mid Z_i(0) = x )     \\
&= P_\infty(Z_i(t) \geq y-x-\epsilon) - P_\infty(Z_i(t) \geq y-x) \\
&= P_\infty(y-x-\epsilon \leq Z_i(t) \leq y-x) \;\;\longrightarrow\; 0,
\quad \text{as } \epsilon \to 0.
\end{aligned}
\]
\end{proof}

\begin{lemma}
\label{lem:bound l_ij} 
The derivative of $l_{ij}(x;b)$ is
$$
l'_{ij}(x; b) = \frac{1 -e^{-(b-x)}(\omega_i + b)}{\lambda_i \cdot D_F(f_i ||g_i)}  -  e^{-(b-x)} L_{ij}(0;b).
$$
Also, $ -l'_{ij}(x; b) $ is bounded by
\begin{equation}
 -l'_{ij}(x; b) \leq e^{-(b-x) }  \bigg(\frac{\omega_i + b}{\lambda_i \cdot D_F(f_i ||g_i)} + L_{ij}(0;b)\bigg).
\end{equation}
\end{lemma}

\medskip

\begin{lemma}
\label{lem:strict_positive_conditional_P} 
 Fix  $i, j \in \mathcal{I}, i \neq j$, choose proper $\lambda_i>0$, s.t., $\mathbb{E}_{\infty} \left[e^{\lambda_i\cdot (S_H(X_{i,k}, f_i) -S_H(X_{i,k}, g_i) )} \right] = 1$,  then for large enough  $b$, the conditional probability
\begin{equation}
      P_{\nu, j}(T_i(b) > T_j(b) \mid T> \nu) >0.   \label{eq:strict_positive} 
\end{equation}
\end{lemma}
Remark:  If \eqref{eq:strict_positive} does not hold, then
$$
P_{\nu, j}(T_i > T_j \mid T> \nu) = \frac{
P_{\nu, j}(T_i > T_j , T> \nu)}{ P_{\nu, j}(T> \nu)}=0 . 
$$
This means the event $\{ T_i > T_j , T> \nu \}$ is is a set with measure 0 under $P_{\nu, j}$. This motivates Lemma~\ref{lem:strict_positive_conditional_P}, 
since otherwise the conditional expectation 
$\mathbb{E}_{\nu,j}[T_i - T_j \mid T>\nu,\, T_i > T_j]$ 
appearing later in \eqref{eq:p_true_isolation} would not be well defined. 
This issue arises in the proof of Theorem~\ref{thm:prob_misisolation} 
(see Section~\ref{proof:prob_upperbound}).
\vspace{0.5em}
\begin{proof}
We follow the proof strategy of \cite{warner2024worst}. 
By the law of total probability,
\begin{align}
   P_{\nu,j}(T_i > T_j \mid T> \nu) & \geq P_{\nu,j}(Z_i(\nu) \leq b/2 \mid T> \nu) \nonumber\\
&  \cdot  P_{\nu,j}(T_i > T_j \mid Z_i(\nu) \leq b/2, T> \nu).   \label{eq:two_parts}
\end{align}
The first factor in \eqref{eq:two_parts} is bounded away from zero by Lemma~\ref{lem:main_condition_holds}, 
$$
P_{\nu,j}(Z_i(\nu) \leq b/2 \mid T> \nu) \geq 1 - e^{-b/2} > 0. 
$$ 
It remains to show the second factor in \eqref{eq:two_parts} is also positive. 

Exactly as in \cite{warner2024worst}, conditioning on $\mathcal{F}_{\nu}$ 
and applying the monotonicity of $Z_i$ and $Z_j$ yields
\begin{equation}
  P_{\nu,j}(T_i > T_j \mid Z_i(\nu)\leq b/2,\, T>\nu) 
    \;\geq\; P_j(T_i(b/2) > T_j(b)).
\label{eq:stochastic_order}
\end{equation}
Since this step is no different from \cite{warner2024worst}, 
we omit the detailed reasoning and only state the result.

From this point onward, our proof departs from \cite{warner2024worst}.  
In particular, we show that the probability in \eqref{eq:stochastic_order} 
is strictly positive. Suppose, for the sake of contradiction, that it is not; 
then we must have
\begin{equation}
P_j( T_i(b/2) \leq T_j(b) ) =1.   
\label{eq:contradiction}
\end{equation} 
Then combine \eqref{eq:contradiction}  and \eqref{eq:E_j_sigma_j}, we can write 
\begin{equation}
\mathbb{E}_j[T_i(b/2)] \leq \mathbb{E}_j[T_j(b)] \leq \frac{b}{\lambda_j {I_j}} + \frac{\mathbb{E}_j\left[ ( U_j^+ )^2 \right]}{  I_j^2 }. \label{eq:contradiction_part1}  
\end{equation}
By (\ref{eq:E_j_sigma_i}), we can write
\begin{equation}
 \mathbb{E}_j[T_i(b/2)] \geq e^{b/2}.   \label{eq:contradiction_part2}  
\end{equation}
Combining (\ref{eq:contradiction_part1}) and (\ref{eq:contradiction_part2}), we then get 
\begin{equation}
 e^{b/2} \leq \mathbb{E}_j[T_j(b)] \leq \frac{b}{\lambda_j {I_j}} + \frac{\mathbb{E}_j\left[ ( U_j^+ )^2 \right]}{  I_j^2 }, \quad     \forall b>0.  \label{eq:contradiction_result}
\end{equation}
Obviously, $\lim\limits_{b \rightarrow \infty}   e^{b/2} - \left(   \frac{b}{\lambda_j {I_j}} + \frac{\mathbb{E}_j\left[ ( U_j^+ )^2 \right]}{  I_j^2 } \right)  >0$. Therefore, (\ref{eq:contradiction_result}) does not hold for sufficiently large $b$.
\end{proof}

\section*{Acknowledgment}
This work was supported by the U.S. National Science Foundation under awards numbers 2334898 and 2334897.










\bibliographystyle{IEEEtran}
\bibliography{bibtex/bib/IEEEabrv,bibtex/bib/IEEEexample}

%







\end{document}